\documentclass[conference, lettepaper]{IEEEtran}
\IEEEoverridecommandlockouts
\usepackage{cite}
\usepackage{amsmath,amssymb,amsfonts,amsthm}
\usepackage{stfloats}
\usepackage{blindtext}
\usepackage{algorithmic}
\usepackage{graphicx}
\usepackage{textcomp}
\usepackage{xcolor}
\def\BibTeX{{\rm B\kern-.05em{\sc i\kern-.025em b}\kern-.08em
    T\kern-.1667em\lower.7ex\hbox{E}\kern-.125emX}}
    
\theoremstyle{theorem}
\newtheorem{theorem}{Theorem}

\newtheorem{proposition}[theorem]{Proposition}

\theoremstyle{definition}
\newtheorem{definition}[theorem]{Definition}
\newtheorem{example}[theorem]{Example}
\newtheorem{remark}[theorem]{Remark}

\newcommand{\C}{\mathcal{C}}

\newcommand{\F}{\mathbb{F}}

\newcommand{\N}{\mathbb{N}}

\newcommand{\wt}{\mathrm{wt}}

\newcommand{\G}{\mathcal{G}}

\newcommand{\dfree}{\mathrm{d_{free}}}

\newcommand{\T}{\mathcal{T}}

\begin{document}

\title{Construction of Rate $(n-1)/n$ Non-Binary LDPC Convolutional Codes via Difference Triangle Sets \\
\thanks{The authors acknowledge the support of the Swiss National Science Foundation grant n. 188430. Julia Lieb acknowledges also the support of the German Research Foundation grant LI 3101/1-1.}
}

\author{\IEEEauthorblockN{Gianira Nicoletta Alfarano}
\IEEEauthorblockA{\textit{University of Zurich}\\
Switzerland \\
gianiranicoletta.alfarano@math.uzh.ch}
\and
\IEEEauthorblockN{Julia Lieb}
\IEEEauthorblockA{\textit{University of Zurich} \\
Switzerland  \\
julia.lieb@math.uzh.ch}
\and
\IEEEauthorblockN{Joachim Rosenthal}
\IEEEauthorblockA{\textit{Fellow, IEEE} \\
\textit{University of Zurich} \\
Switzerland  \\
rosenthal@math.uzh.ch}

}

\maketitle

\begin{abstract}
This paper provides a construction of non-binary LDPC convolutional codes, which generalizes the work of Robinson and Bernstein. The sets of integers forming an $(n-1,w)$-difference triangle set are used as supports of the columns of rate $(n-1)/n$ convolutional codes. If the field size is large enough, the Tanner graph associated to the sliding parity-check matrix of the code is free from $4$ and $6$-cycles not satisfying the full rank condition. This is important for improving the performance of a code and avoiding the presence of low-weight codewords and absorbing sets. The parameters of the convolutional code are shown to be determined by the parameters of the underlying difference triangle set. In particular, the free distance of the code is related to $w$ and the degree of the code is linked to the ``scope" of the difference triangle set. Hence, the problem of finding families of difference triangle set with minimum scope is equivalent to find convolutional codes with small degree. 
\end{abstract}


\section{Introduction}
The aim of this paper is to construct a family of non-binary low-density parity-check (NB-LDPC) convolutional codes suitable for iterative deoding. The class of LDPC block codes was introduced by Gallager \cite{gallager1962low}. Their name is due to the fact that they have a parity-check matrices that is sparse. Similarly to LDPC block codes, one can construct  LDPC convolutional codes as codes whose sliding parity-check matrices are sparse, which allows them to be decoded using iterative message-passing algorithms. 

In the last few years, some attempts to construct binary LDPC convolutional codes were done. However, most of the constructions are for time-varying convolutional codes, see for instance \cite{zhou2010cycle, pusane2011deriving, battaglioni2019girth}.

In 1967, Robinson and Bernstein \cite{robinson1967class} used  difference triangle sets for the first time to construct binary recurrent codes, which are defined as the kernel of a binary sliding parity-check matrix. At that time, the theory of convolutional codes was not developed yet and the polynomial notation was not diffused, but now, we may regard recurrent codes as a first version of convolutional codes. This was the first time that a combinatorial object was used to construct convolutional codes. Three years later, Tong in \cite{tong1970systematic}, used diffuse difference triangle sets to construct self-orthogonal diffuse convolutional codes, defined by Massey \cite{massey1963threshold}.  The aim of these authors was to construct codes suitable for iterative decoding and their result was a rudimental version of binary LDPC convolutional codes. 

In this paper, we exploit the structure of difference triangle sets to construct non-binary LDPC convolutional codes, whose parity check matrices are free from 4-cycles and 6-cycles not satisfying the so called full rank condition. Our construction may be regarded as a generalization over $\F_q$ of the construction of Robinson and Bernstein. We describe a close link between the properties of the difference triangle set and the parameters of the code. Moreover, we derive information on the column distances and on the free distance of the constructed codes, by exploiting the structure of the underlying difference triangle set.

The paper is structured as follows. In Section \ref{sec:preliminaries}, we first give some useful basics of the theory of convolutional codes and then we define difference triangle sets and their scope. In Section \ref{sec:LDPC}, we define non-binary LDPC block codes and non-binary LDPC convolutional codes.  In Section \ref{sec:construction}, we give a new construction of rate $(n-1)/n$ non-binary LDPC convolutional codes, starting from an $(n-1,w)$ difference triangle set. We show how the parameters of the code are related to the properties of the triangle set and we point out that several research works in combinatorics can be exploited to improve our construction. We derive some distance properties of the codes and the exact formula for computing their density. We conclude with further comments and future research directions in Section \ref{sec:Conclusion}.

\section{Preliminaries}\label{sec:preliminaries}

\subsection{Convolutional Codes}
Let $q$ be a prime power, $\F_q$ be the finite field of order $q$ and $k,n$ be positive integers, with $k\leq n$. A rate-$k/n$ convolutional code over $\F_q$ is a submodule $\C$ of $\F_q[z]^n$ of rank $k$, such that there exists a $k\times n$ polynomial generator matrix $G(z)\in \F_q[z]^{k\times n}$ which is \emph{basic} and \emph{reduced}, i.e., it has a right polynomial inverse and  the sum of the row degrees of $G(z)$ attains the minimal possible value such that
$$\C:=\{u(z)G(z)\mid u(z) \in \F_q[z]^k\}\subseteq \F_q[z]^n.$$
If $G(z)$ is a reduced, basic generator matrix for $\C$, there exists a \emph{parity-check} matrix $H(z)\in\F_q[z]^{(n-k)\times n}$ with $H_0$ full rank such that
$$\C := \{v(z)\in\F_q[z]^n \mid H(z)v(z)^\top={0}\}.$$
We define the \emph{degree} $\delta$ of $\C$ as the highest degree of the $k\times k$ full size minors in $G(z)$. We denote a convolutional code  of rank $k/n$ and degree $\delta$ by $(n,k,\delta)_q$. 
For a polynomial vector $v(z)=\sum_{i=0}^r v_iz^i\in \C$, we define the \emph{weight} of $v(z)$ as  $\wt(v(z)) := \sum _{i=0}^r \wt(v_i)\in\N_0,$ where $\wt(v_i)$ denotes the Hamming weight of $v_i\in\F_q^n$.
The \emph{free distance}  of a convolutional code $\C$, $\dfree(\C)$, is defined as the minimum of the nonzero weights of the codewords in $\C$.
The parameters $\delta$ and $\dfree$ are needed to determine respectively the decoding complexity and the error correction capability of a convolutional code with respect to some decoding algorithm. For this reason, for any given rate $k/n$ and field size $q$, the aim is to construct convolutional codes with ``small'' degree $\delta$ and  ``large'' free distance $\dfree$. 

\begin{remark}
There is a natural isomorphism between $\F_q[z]^n$ and $\F_q^n[z]$ that allows to consider a generator and a parity-check matrix of a convolutional code as polynomials whose coefficients are matrices. In particular, we will consider $H(z)\in\F_q^{(n-k)\times n}[z]$, such that $H(z) = H_0+H_1z +\dots H_{\mu}z^\mu$, with $\mu>0$. With this notation, we can expand the kernel representation $H(z)v(z)^\top$ in the following way:
\begin{equation}\label{eq:kerH}
Hv^\top = \begin{bmatrix}
H_0 & & & &  \\
\vdots & \ddots & & &  \\
H_{\mu} & \cdots & H_0 & &  \\
 & \ddots &  &\ddots &  \\
 & & H_{\mu} &\cdots & H_0 \\
 & & & \ddots & \vdots \\
 & & & & H_\mu
\end{bmatrix}\begin{bmatrix}
v_0 \\ v_1 \\ \vdots \\ v_r
\end{bmatrix}=0
\end{equation}
\end{remark}

We will refer to the representation of the parity-check matrix of $\C$ in equation \eqref{eq:kerH} as \emph{sliding parity-check matrix}. 

For any $j \in \N_0$ we define the \emph{$j$-th column distance of $\C$} as
\begin{align*}
d_j^c(\C)&:= \min_{v_0\neq 0} \biggl\{\wt(v_0 +v_1z + \dots +v_jz^j) \mid v(z) 
\in \C\biggr\}\\ 
&=\min_{v_0\neq 0} \biggl\{\wt(v_0 + \dots +v_jz^j) \mid H_j^c[v_0\cdots v_j]^{\top}=0 \biggr\}
\end{align*}

with $H_j^c:=\begin{bmatrix}
    H_0 & & & \\
    H_1 & H_0 & &  \\
    \vdots & \vdots & \ddots & \\
    H_{j} & H_{j-1} & \cdots & H_0
    \end{bmatrix}$.
    
We recall the following result.
\begin{theorem}\cite[Proposition 2.2]{gluesing2006strongly}\label{cd}
Let $d\in\N$. Then the following properties are equivalent.
\begin{enumerate}
\item $d_j^c = d$.
\item None of the first $n$ columns of $H_j^c$ is contained in the span of any other $d -2$ columns and one of the first $n$ columns of  $H_j^c$  is in the span of some other $d-1$ columns of that matrix.

\end{enumerate} 
\end{theorem}

\subsection{Difference Triangle Sets}
A difference triangle set is a collection of sets of integers such that any integer can be written in at most one way as difference of two elements in the same set. Difference triangle sets find application in combinatorics, radio systems, optical orthogonal codes and other areas of mathematics \cite{klove1989bounds, chee1997constructions, chen1992disjoint}. We refer to \cite{colbourn1996difference} for a more detailed treatment. More formally, we define difference triangle sets in the following way.

\begin{definition}
An $(N,M)$-\emph{difference triangle set} (DTS) is a set $\T:=\{T_1, T_2, \dots, T_N\}$, where for any $1\leq i\leq N$, $T_i:=\{a_{i,j} \mid 1\leq j\leq M\}$ is a set of nonnegative integers such that $a_{i,1} <a_{i,2} < \cdots <a_{i,M}$ and all the differences $a_{i,j}-a_{i,k}$, with $1\leq i\leq N$ and $1\leq k <j \leq M$ are distinct. When $N =1$, we will refer to a $(1,M)$-DTS simply as DTS. 
\end{definition}

An important parameter characterizing an $(N,M)$-DTS $\T$ is the \emph{scope} $m(\T)$, that is defined as
$$ m(\T):= \max\{a_{i,M} \mid 1\leq i \leq N\}.$$ 

Observe that, a very well-studied problem in combinatorics is finding families of $(N,M)$-DTSs with minimum scope. In this work, we will use the sets in a DTS as supports of the columns in the sliding parity-check matrix of a convolutional code. We will relate the scope of the DTS with the degree of the code. Since we want to minimize the degree of the code, it is evident that the mentioned combinatorial problem plays a crucial role also here.

\section{Low-Density Parity-Check Codes}\label{sec:LDPC}

\subsection{Non-Binary LDPC Codes}
In this section we briefly introduce LDPC block codes and we focus in particular on their non-binary version. We extend then the notion to LDPC convolutional codes.

LDPC codes are known for their performances near the Shannon-limit over the additive white Gaussian noise channel \cite{mackay1996near}. Their non-binary (NB-LDPC) version was first investigated by Davey and Mackay in 1998 in \cite{davey1998low}. In \cite{davey1998monte}, it was observed that NB-LDPC codes defined over a finite field with $q$ elements can have better performances than the binary ones. A NB-LDPC code is defined as the kernel of an $N\times M$ sparse (at least 1/2 of the entries are zeros) matrix $H$ with entries in $\F_q$. We can associate to $H$ a bipartite graph $\G = (V,E)$, called \emph{Tanner graph}, where $V = V_s \cup V_c$ is the set of vertices. In particular, $V_s=\{v_1,\dots, v_N\}$ is the set of \emph{variable nodes} and $V_c=\{c_1,\dots,c_M\}$ is the set of \emph{check nodes}. $E\subseteq V_s\times V_c$ is the set of edges, with $e_{n,m}=(v_n,c_m)\in E$ if and only if $h_{n,m}\ne 0$. The edge $e_{n,m}$ connecting a check node and a variable node is labelled by $h_{n,m}$, that is the  corresponding \emph{permutation node}. For an even integer $\ell$, we call a simple closed path consisting of $\ell/2$ check nodes and $\ell/2$ variable nodes in $\G$ an \emph{$\ell$-cycle}. The length of the shortest cycle is called the \emph{girth} of $\G$ or girth of $H$. It is proved that having higher girth decreases the decoding failure of the bit flipping algorithm. Moreover, in \cite{poulliat2008design} the authors showed that short cycles in a NB-LDPC code may be harmful if they do not satisfy the so called full rank condition (FRC). This is because if the FRC is not satisfied, the short cycles produce low-weight codewords or they form absorbing sets, \cite{amiri2014analysis}.

In \cite{poulliat2008design} and in \cite{amiri2014analysis} it is shown that an $\ell$-cycle in a NB-LDPC code with parity check matrix $H$ can be represented by an $\frac{\ell}{2}\times \frac{\ell}{2}$ submatrix of $H$ of the form


\begin{equation}\label{eq:cycles}
A=\begin{bmatrix}
a_1 & a_2 & 0 & \cdots &\cdots & 0 \\
0 & a_3 & a_4 & \cdots  &\cdots & \vdots\\
\vdots & & \ddots & & & \vdots \\
\vdots & & & \ddots & & \vdots\\
0 & & & & a_{\ell-3} & a_{\ell-2} \\
a_\ell & 0 & \cdots & \cdots &0 & a_{\ell-1}
\end{bmatrix},
\end{equation}
where $a_i\in\F_q^\ast$. The cycle does not satisfy the FRC if $\det(A)=0$. In this case, the cycle gives an absorbing set. Hence, it is a common problem to construct NB-LDPC codes in which the shortest cycles satisfy the FRC.

The convolutional counterpart of NB-LDPC block codes is given by convolutional codes defined over a finite field $\F_q$ whose sliding parity-check matrix is sparse.

\section{Construction of Rate $(n-1)/n$ NB-LDPC convolutional Codes}\label{sec:construction}
In this section we will provide a construction of NB-LDPC convolutional codes over $\F_q$, with the aid of difference triangle sets. In a certain sense, this could be regarded as an extension over $\F_q$ of the construction given by Robinson and Bernstein.

Let $\F_q$ be the finite field of order $q=p^N$, where $p$ is a prime number. 

We are going to construct a sliding parity-check matrix as in equation \eqref{eq:kerH}. Observe that the decoding of a convolutional code $\C$ is done sequentially by blocks of length $n$, hence, the error-correcting properties of the code are determined by the decoding of the first block (see also \cite{wyner1963analysis}). In particular, it is sufficient to analyze the portion of the sliding parity-check matrix $H$ which affects the decoding of the first block, namely
\begin{equation}\label{eq:slidingportion}
  \mathcal{H}:= H_{\mu}^c=\begin{bmatrix}
    H_0 & & & \\
    H_1 & H_0 & &  \\
    \vdots & \vdots & \ddots & \\
    H_{\mu} & H_{\mu-1} & \cdots & H_0
    \end{bmatrix}.
\end{equation}

First of all, observe that since $H_0$ is full rank, one can perform Gaussian elimination on the block $\begin{bmatrix}
H_0^\top & H_1^\top & \cdots & H_{\mu}^\top
\end{bmatrix}^\top$,
which results in the following block matrix:
\begin{equation}\label{eq:matrixH}
\bar{H}=\begin{bmatrix}
    A_0 & | & I_{n-k}  \\
    A_1 & | & 0 \\
    \vdots & & \vdots \\
    A_{\mu} & | & 0
\end{bmatrix},\end{equation}
where $A_i\in \F_q^{(n-k)\times k}$ for $i=1,\hdots, \mu$. 
With an abuse of notation, we will still write $H_0$ for indicating  $[A_0|I_{n-k}]$, and $H_i$ for the matrices $[A_i | 0]$.

Note that it is important to construct the sliding parity-check matrix $H$ of a NB-LDPC convolutional code such that the Tanner graph $\mathcal{G}$ associated to $H$ does not contain short cycles not satisfying the FRC. It is easy to see that $H$ satisfies this property if and only if $\mathcal{H}$ does. By the discussion of the previous section, this is equivalent to construct $\mathcal{H}$, such that all the $2\times 2$ and $3\times 3$ minors that are non-trivially zero, are non-zero.

In the following we focus on the construction of rate $(n-1)/n$ NB-LDPC convolutional codes. In particular, we will construct the matrices $A_i\in \F_q^{1\times (n-1)}$, such that the resulting matrix $\mathcal{H}$ does not contain $4$-cycles and $6$-cycles, not satisfying the FRC.

\subsection{Construction}
Let $n,w$ be positive integers. Consider an $(n-1, w)$-DTS $\T:=\{T_1, \dots, T_{n-1}\}$. 
Each $T_k$ will give the positions of the non-zero elements of the first $n-1$ columns of the matrix $\bar{H}$ of equation \eqref{eq:matrixH}; the last column will be simply given by the vector $[1,0,\dots, 0]^\top$.

\begin{definition}\label{Construction}
With the notation above, define the matrix $\bar{H}^{\T}\in\F_q^{m(\T)\times n}$, in which 
the $k$-th column has weight $w$ and support $T_k:=\{a_{k,1}, \dots, a_{k,w}\}$. Formally, let $\alpha$ be a primitive element for $\F_q$, so that any non-zero element of $\F_q$ can be written as power of $\alpha$. For any $1\leq i \leq m(\T)$, $1\leq k\leq n-1$, 
$$
\bar{H}^\T_{i,k} = \begin{cases}\alpha^{ik} & \text{ if } i \in T_k \\
0 & \text{ otherwise}
\end{cases}.$$
The last column of $\bar{H}^\T$ is given by $[1,0,\cdots,0]^\top$.
Derive  the matrix $\mathcal{H}^\T$ by ``shifting" the columns of $\bar{H}^\T$ and then a sliding matrix $H^\T$ of the form of equation \eqref{eq:kerH}. Finally, define $\C^\T := \ker(\mathcal{H}^\T)$ over $\F_q$.
Note that here $\mu = m(\T)-1$.

\end{definition}

\begin{example}\label{ex:Construction}
Let $\F_q:=\{0,1,\alpha, \dots, \alpha^{q-2}\}$ 
and $\T$ be a $(2,3)$-DTS, such that $T_1:=\{1,2,6\}$ and $T_2:=\{1,2,4\}$. 
Then, with the notation above,
$$\bar{H}^\T= \begin{bmatrix}
\alpha & \alpha^{2} & 1 \\
\alpha^{2} & \alpha^{4} & 0 \\
0 & 0 & 0 \\
0 & \alpha^6 & 0 \\
0 & 0 & 0 \\
\alpha^6 & 0 & 0 
\end{bmatrix},$$
which leads to the sliding matrix in Figure \ref{fig:slid1}.

\begin{figure*}[b!]
\centering
$$\mathcal{H}^\T= \left[\begin{array}{cccccccccccccccccc}
\alpha & \alpha^{2} & 1 & & & & & & & &  \\

\alpha^{2} & \alpha^{4} & 0 & \alpha & \alpha^2 & 1 & & & & & & & & & & & \\

0 & 0 & 0  & \alpha^{2} & \alpha^{4} & 0  & \alpha & \alpha^2 & 1 & & & & & & & & \\

0 & \alpha^8 & 0 & 0 & 0 & 0  & \alpha^{2} & \alpha^{4} & 0  & \alpha & \alpha^2 & 1 & & & & & \\

0 & 0 & 0 &  0 & \alpha^8 & 0&  0 & 0 & 0  & \alpha^{2} & \alpha^{4} & 0  & \alpha & \alpha^2 & 1 & & &  \\

\alpha^6 & 0 & 0 & 0 & 0 & 0 &  0 & \alpha^8 & 0 & 0 & 0 & 0  & \alpha^{2} & \alpha^{4} & 0  & \alpha & \alpha^2 & 1 \\
\end{array}\right]$$
\caption{Sliding parity-check matrix for the code in Example \ref{ex:Construction}.\label{fig:slid1}}
\end{figure*}
\end{example}

\begin{example}\label{ex:2}
Let $\F_q:=\{0,1,\alpha, \dots, \alpha^{q-2}\}$ 
and $\T$ be a $(2,3)$-DTS, such that $T_1:=\{1,2,6\}$ and $T_2:=\{2,3,5\}$. 
Then, with the notation above,
$$\bar{H}^\T= \begin{bmatrix}
\alpha & 0 & 1 \\
\alpha^{2} & \alpha^{4} & 0 \\
0 & \alpha^6 & 0 \\
0 & 0 & 0 \\
0 & \alpha^{10} & 0 \\
\alpha^6 & 0 & 0 
\end{bmatrix},$$
which leads to the sliding matrix in Figure \ref{fig:2}.
\begin{figure*}[b!]
$$\mathcal{H}^\T= \left[\begin{array}{cccccccccccccccccc}
\alpha & 0 & 1 & & & & & & & &  \\

\alpha^{2} & \alpha^{4} & 0 & \alpha & 0 & 1 & & & & & & & & & & & \\

0 & \alpha^6 & 0  & \alpha^{2} & \alpha^{4} & 0  & \alpha & 0 & 1 & & & & & & & & \\

0 & 0 & 0 & 0 & \alpha^6 & 0  & \alpha^{2} & \alpha^{4} & 0  & \alpha & 0 & 1 & & & & & \\

0 & \alpha^{10} & 0 &  0 & 0 & 0&  0 & \alpha^6 & 0  & \alpha^{2} & \alpha^{4} & 0  & \alpha & 0 & 1 & & &  \\

\alpha^6 & 0 & 0 & 0 & \alpha^{10} & 0 &  0 & 0 & 0 & 0 & \alpha^6 & 0  & \alpha^{2} & \alpha^{4} & 0  & \alpha & 0 & 1 \\
\end{array}\right]
$$
\caption{Sliding parity-check matrix for the code in Example \ref{ex:2}.\label{fig:2}}
\end{figure*}
\end{example}

\begin{proposition}
Let $\T$ be an $(n-1,w)$-DTS with scope $m(\T)$. Then, the code $\C^\T$ given as in Definition \ref{Construction} is an $(n,n-1,m(\T)-1)_q$ convolutional code. 
\end{proposition}


\begin{remark}
As already mentioned, an interesting problem in combinatorics is to find  families of difference triangle sets having minimum scope \cite{chee1997constructions, klove1989bounds, colbourn1996difference}. This is a difficult task in general. For our application, it is desirable to have a difference triangle set $\T$ whose scope is as small as possible so that the degree of $\C^\T$ is small as well. This is desirable for convolutional codes because the complexity of the decoding algorithm increases with $\delta$. 
\end{remark}

\begin{theorem}\label{distance}
Let $\mathcal{T}$ be an $(n-1,w)$-DTS and consider the matrix 
$\begin{bmatrix}
    A_0^\top & \cdots &    A_{\mu}^\top
\end{bmatrix}^\top$ defined as in the previous construction.
Denote by $w_j$ the minimal column weight of $\begin{bmatrix}
    A_0^\top & \cdots &    A_{j}^\top
\end{bmatrix}^\top$. For $I\subset\{1,\hdots,\mu+1\}$ and $J\subset\{1,\hdots,n(\mu+1)\}$ we define  $[\mathcal{H}^\T]_{I;J}$ as the submatrix of $\mathcal{H}^\T$ with row indices $I$ and column indices $J$. Assume that for all $I, J$ with $|I|=|J|\leq w$ and $j_1:=\min(J)\leq n-1$ and $I$ containing the indices where column $j_1$ is nonzero, we have that the first column of $[\mathcal{H}^\T]_{I;J}$ is not contained in the span of the other columns of $[\mathcal{H}^\T]_{I;J}$. Then 
\begin{itemize}
    \item[(i)] $\dfree(\C^\T)=w+1$,
    \item[(ii)] $d_j^c=w_j+1$.
\end{itemize}
\end{theorem}


\begin{proof}
(i) Without loss of generality, we can assume that the first entry of $H_0$ is nonzero. Let $M\subset\{1,\hdots,\delta+1\}$ with $|M|=w$ be the set of positions where the first column of $\mathcal{H}$ (and hence also the first column of the sliding parity-check matrix) has nonzero entries. Denote the values of these nonzero entries by $d_1,\hdots, d_w$. Then, $v(z)=\sum_{i=0}^r v_iz^i$ with $v_0=[1\ 0\cdots 0\ -d_1]$ and $v_i=\begin{cases}[0\cdots 0]& \text{for}\ i+1\notin M\\ [ 0\cdots 0\ -d_{i+1}]& \text{for}\ i+1\in M\end{cases}$ for $i\geq 1$ is a codeword with $\wt(v(z))=w+1$.
Hence $\dfree\leq w+1$.

Assume by contradiction that there exists a codeword $v\neq 0$ with weight $d\leq w$. We can assume that $v_0\neq 0$, i.e. there exists $i\in\{1,\ldots,n\}$ with $v_{0,i}\neq 0$. One knows $\mathcal{H}^\T v^{\top}=0$. Of this homogeneous system of equations, where we consider the nonzero components of $v_0, v_1, \ldots, v_{\deg(v)}$ as variables, we take only the rows where column $i$ of $\mathcal{H}^\T$ has nonzero entries. We end up with a homogeneous system with $w$ equations and $d$ variables, whose coefficient matrix has full column rank according to the assumptions of the theorem. This implies $v=0$, what is a contradiction.\\
(ii) The result follows from Theorem \ref{cd} with an analogue reasoning as in part (i).
\end{proof}

\begin{remark}\label{rem:dist}
With the assumptions of Theorem \ref{distance}, one has $d_j^c=\dfree(\C^\T)$ for $j\geq\mu$. Moreover, one achieves higher column distances (especially for small $j$) if the elements of $\mathcal{T}$ are small.
\end{remark}

\begin{proposition}\label{density}
If $N$ is the maximal message length, i.e. for any message $v$, $\deg(v)+1\leq N/n$, then the sliding parity-check matrix of a convolutional code derived in Definition \ref{Construction} has density 
$$\frac{w(n-1)+1}{\mu n+N}.$$
\end{proposition}

\begin{proof}
To compute the density of a matrix, one has to divide the number of nonzero entries by the total number of entries. The result follows immediately.
\end{proof}

\begin{theorem}\label{thm:2x2minors}
Let $\T$ be an $(n-1,w)$-DTS with scope $m(\T)$ and $\F_q$ be the finite field with $q$ elements with $q>(n-1)\delta+1=(n-1)(m(\T)-1)+1$. Let $\C^\T$ be the rate $(n-1)/n$ convolutional code defined over $\F_q$ from $\T$, with $\mathcal{H}^\T$ as defined in \eqref{eq:slidingportion}. 
Then, all the $2\times 2$ minors in $\mathcal{H}^\T$ that are non-trivially zero are non-zero.
\end{theorem}
\begin{proof}
The only $2\times 2$ minors to check are the ones of the form $\begin{vmatrix}a_1 & a_2\\ a_3 & a_4
\end{vmatrix}$. By definition of DTS, the support of any column of $\mathcal{H}^\T$ intersects the support of its shift at most once. This ensures that the columns of all these minors are the shift of two different columns of $\bar{H}^\T$. Moreover, all the elements in the minor are powers of $\alpha$. In particular, let $1\leq i,r \leq \delta$, $0\leq j,k \leq n-1$ (note that $j<k$ or $k<j$ according to which columns from $\bar{H}^\T$ are involved in the shifts). Hence we have that:
\begin{align*}
& \begin{vmatrix}a_1 & a_2\\ a_3 & a_4
 \end{vmatrix}  =
 \begin{vmatrix}\alpha^{ij} & \alpha^{lk}\\ \alpha^{(i+r)j} & \alpha^{(l+r)k}
\end{vmatrix}  = \\
&\alpha^{ij}\alpha^{(l+r)k} - \alpha^{lk}\alpha^{(i+r)j} =
\alpha^{ij + lk}(\alpha^{rk}-\alpha^{rj})
\end{align*}
which is $0$ if and only if $rk = rj \mod (q-1)$. Since it holds that $0\leq j < k \leq n-1$ or $0\leq k < j \leq n-1$  and $1\leq r \leq \delta$, this can not happen.
\end{proof}


\begin{theorem}\label{thm:3x3minors}
Let $\T$ be an $(n-1,w)$-DTS with scope $m(\T)$, $w\geq 3$ and $\F_q$ be the finite field with $q>2$ elements with $q=p^N$, where $N>(\delta-1) (n-2)=(m(\T)-2)( n-2)$. Let $\C^\T$ be the rate $(n-1)/n$ convolutional code defined over $\F_q$ from $\T$, with $\mathcal{H}^\T$ as defined in \eqref{eq:slidingportion}. 
Then, all the $3\times 3$ minors in $\mathcal{H}^\T$ that are non-trivially zero are non-zero.
\end{theorem}
\begin{proof}
We need to distinguish different cases. 

\underline{\textbf{Case I}}. The $3\times 3$ minors are of the form $\begin{vmatrix}a_1 & a_2 & a_3\\ a_4 & a_5 & a_6 \\ a_7 & a_8 & a_9
\end{vmatrix}$, with $a_i \ne 0$ for any $i$. As we observed in Theorem \ref{thm:2x2minors}, in this case all the columns are shifts of three different columns from $\bar{H}^\T$. Hence we have that, given $1\leq i,l,t \leq \delta-3$, $r,s>0$, with $r\ne s$ and $2\leq i+r, l+r, t+r\leq \delta - 1$ and $4\leq i+r+s, l+r+s, t+r+s \leq \delta$, the minors are given by
\begin{gather*}
    \begin{vmatrix}a_1 & a_2 & a_3\\ a_4 & a_5 & a_6 \\ a_7 & a_8 & a_9
\end{vmatrix} = \begin{vmatrix}
\alpha^{ij} & \alpha^{lk} & \alpha^{tm} \\
\alpha^{(i+r)j} & \alpha^{(l+r)k} & \alpha^{(t+r)m}\\
\alpha^{(i+r+s)j} & \alpha^{(l+r+s)k} & \alpha^{(t+r+s)m}\\
\end{vmatrix}.
\end{gather*}
This determinant is $0$ if and only if 
\begin{gather}\label{eq:determinant3x3}
\alpha^{rk+rm+sm} + \alpha^{rm+rj+sj}+ \alpha^{rj+rk+sk}=\\
\alpha^{rk+rj+sj}+ \alpha^{rj+rm+sm}+ \alpha^{rk+rm+sk}.    
\end{gather}
Without loss of generality we can assume that $j<k<m$ and it turns out that the maximum exponent in equation \eqref{eq:determinant3x3} is $rk+rm+sm$ while the minimum is $rk + rj + sj$. Let $M:=rk+rm+sm - (rk + rj + sj)$. We immediately see that the maximum value for $M$ is $(\delta-1)(n-2)$ hence this determinant can not be zero because $\alpha$ is a primitive element for $\F_q$ and, by assumption, $q=p^N$, where $N>M$. 

\underline{\textbf{Case II}}. The $3\times 3$ minors are of the form $\begin{vmatrix}a_1 & a_2 & 0\\ 0 & a_3 & a_4 \\ a_6 & 0 & a_5
\end{vmatrix}$. Arguing as before, we notice that given $1\leq i,l,t \leq \delta-3$, $r,s>0$, with $r\ne s$ and $2\leq i+r, l+r, t+r\leq \delta - 1$ and $4\leq i+r+s, l+r+s, t+r+s \leq \delta$, the minors are given by
\begin{gather*}
    \begin{vmatrix}
\alpha^{ij} & \alpha^{lk} & 0 \\
0 & \alpha^{(l+r)k} & \alpha^{(t+r)m}\\
\alpha^{(i+r+s)j} & 0 & \alpha^{(t+r+s)m}\\
\end{vmatrix} =\\ \alpha^{ij+lk+tm+rm}(\alpha^{rk+sm}+\alpha^{rj+sj}).
\end{gather*}
This determinant is $0$ whenever $r(k-j) + s(m-j) - (q-1)/2=0 \mod (q-1)$. If $q>2(n-3) + 2(\delta-2)(n-2)+1$ this never happens.  
And this is the case for our field size assumption.

\underline{\textbf{Case III}}. The $3\times 3$ minors are of the form $\begin{vmatrix}a_1 & a_2 & 0\\  a_3 & a_4 & a_5\\ a_6 & 0 & a_7
\end{vmatrix}$. As in the first cases, we can assume that,  for $1\leq i,l,t \leq \delta-3$, $r,s>0$, with $r\ne s$ and $2\leq i+r, l+r, t+r\leq \delta - 1$ and $4\leq i+r+s, l+r+s, t+r+s \leq \delta$, the minor is given by 
\begin{gather*}
    \begin{vmatrix}
\alpha^{ij} & \alpha^{lk} & 0 \\
\alpha^{(i+r)j} & \alpha^{(l+r)k} & \alpha^{(t+r)m}\\
\alpha^{(i+r+s)j} & 0 & \alpha^{(t+r+s)m}\\
\end{vmatrix}.
\end{gather*}
By  following  the  reasoning  of  the  previous  cases, if $N>(\delta-1)(n-2)-1$, this determinant is nonzero which is always the case, because of the field size assumption.
\end{proof}

\begin{example}
In Example \ref{ex:Construction}, one has $d_0^c=2$, $d_1^c=d_2^c=d_3^c=d_4^c=3$ and $d_5=\dfree=4$.
\end{example}

\begin{example}
In Example \ref{ex:2}, one has $d_0^c=1$, $d_1^c=2$, $d_2^c=d_3^c=d_4^c=3$ and $d_5=\dfree=4$.
\end{example}

\begin{remark}
With Theorems \ref{thm:2x2minors} and \ref{thm:3x3minors} we can ensure that the $4$ and $6$-cycles in the Tanner graph associated to codes $\C^\T$ defined over $q = p^N$, with $N>(\delta-1) (n-2)$ satisfy the FRC. This improves the performances of our NB-LDPC convolutional codes. 

Moreover, it is possible to reduce the required field size for the construction of $\C^\T$ by restricting the conditions on the DTS $\T$ and still ensuring that all the  $4$ and $6$-cycles satisfy the FRC. In particular, we can get rid of the \emph{Case I} of Theorem \ref{thm:3x3minors} by imposing that the sets in $\T$ pairwise intersect at most twice and also the support of one column intersects the support of the shifts of any column at most twice, to ensure that all columns of $\mathcal{H}^\T$ intersect at most twice. We will leave these considerations for future works.
\end{remark}

\section{Conclusion and Future Research Works}\label{sec:Conclusion}
In this paper, we gave a construction of rate $(n-1)/n$ convolutional codes over non-binary fields, generalizing a construction from Robinson and Bernstein, using difference triangle sets. We related the important parameters of the codes with the parameters of the considered DTS, pointing out how combinatorics can help in solving applied problems (in this case minimizing the degree $\delta$ of the code). 

Generalizations of this work will be addressed in an extended version. In particular, minors of $\mathcal{H}^\T$ of larger size than $3\times 3$ could be considered to derive convolutional codes with larger distances. Unfortunately, this may require a larger field size.

Moreover, Theorem \ref{distance}, Remark \ref{rem:dist} and Theorem \ref{density} can be generalized to arbitrary  rates $k/n$. However, it is not completely trivial anymore to compute the degree $\delta$ with the help of the parity-check matrix of the code.



\newpage
\bibliographystyle{abbrv}

\bibliography{references}

\begin{thebibliography}{10}

\bibitem{amiri2014analysis}
B.~Amiri, J.~Kliewer, and L.~Dolecek.
\newblock Analysis and enumeration of absorbing sets for non-binary graph-based
  codes.
\newblock {\em IEEE Transactions on Communications}, 62(2):398--409, 2014.

\bibitem{battaglioni2019girth}
M.~Battaglioni, M.~Baldi, F.~Chiaraluce, and M.~Lentmaier.
\newblock Girth properties of time-varying {SC-LDPC} convolutional codes.
\newblock In {\em 2019 IEEE International Symposium on Information Theory
  (ISIT)}, pages 2599--2603. IEEE, 2019.

\bibitem{chee1997constructions}
Y.~M. Chee and C.~J. Colbourn.
\newblock Constructions for difference triangle sets.
\newblock {\em IEEE Transactions on Information Theory}, 43(4):1346--1349,
  1997.

\bibitem{chen1992disjoint}
Z.~Chen, P.~Fan, and F.~Jin.
\newblock Disjoint difference sets, difference triangle sets, and related
  codes.
\newblock {\em IEEE Transactions on Information Theory}, 38(2):518--522, 1992.

\bibitem{colbourn1996difference}
C.~J. Colbourn.
\newblock Difference triangle sets.
\newblock {\em Chapter in The CRC Handbook of Combinatorial Designs by CJ
  Colbourn and J. Dintz}, pages 312--317, 1996.

\bibitem{davey1998low}
M.~C. Davey and D.~J. MacKay.
\newblock Low density parity check codes over {GF} ($q$).
\newblock In {\em 1998 Information Theory Workshop (Cat. No. 98EX131)}, pages
  70--71. IEEE, 1998.

\bibitem{davey1998monte}
M.~C. Davey and D.~J. MacKay.
\newblock Monte {C}arlo simulations of infinite low density parity check codes
  over {GF}($q$).
\newblock In {\em Proc. of Int. Workshop on Optimal Codes and related Topics},
  pages 9--15. Citeseer, 1998.

\bibitem{gallager1962low}
R.~Gallager.
\newblock Low-density parity-check codes.
\newblock {\em IRE Transactions on Information Theory}, 8(1):21--28, 1962.

\bibitem{gluesing2006strongly}
H.~Gluesing-Luerssen, J.~Rosenthal, and R.~Smarandache.
\newblock Strongly-{MDS} convolutional codes.
\newblock {\em IEEE Transactions on Information Theory}, 52(2):584--598, 2006.

\bibitem{klove1989bounds}
T.~Klove.
\newblock Bounds and construction for difference triangle sets.
\newblock {\em IEEE Transactions on Information Theory}, 35(4):879--886, 1989.

\bibitem{mackay1996near}
D.~J. MacKay and R.~M. Neal.
\newblock Near shannon limit performance of low density parity check codes.
\newblock {\em Electronics letters}, 32(18):1645--1646, 1996.

\bibitem{massey1963threshold}
J.~L. Massey.
\newblock Threshold decoding.
\newblock 1963.

\bibitem{poulliat2008design}
C.~Poulliat, M.~Fossorier, and D.~Declercq.
\newblock Design of regular $(2, d_c)$-{LDPC} codes over {GF}($q$) using their
  binary images.
\newblock {\em IEEE Transactions on Communications}, 56(10):1626--1635, 2008.

\bibitem{pusane2011deriving}
A.~E. Pusane, R.~Smarandache, P.~O. Vontobel, and D.~J. Costello.
\newblock Deriving good {LDPC} convolutional codes from {LDPC} block codes.
\newblock {\em IEEE Transactions on Information Theory}, 57(2):835--857, 2011.

\bibitem{robinson1967class}
J.~P. Robinson and A.~Bernstein.
\newblock A class of binary recurrent codes with limited error propagation.
\newblock {\em IEEE Transactions on Information Theory}, 13(1):106--113, 1967.

\bibitem{tong1970systematic}
S.-Y. Tong.
\newblock Systematic construction of self-orthogonal diffuse codes.
\newblock {\em IEEE Transactions on Information Theory}, 16(5):594--604, 1970.

\bibitem{wyner1963analysis}
A.~Wyner and R.~Ash.
\newblock Analysis of recurrent codes.
\newblock {\em IEEE Transactions on Information Theory}, 9(3):143--156, 1963.

\bibitem{zhou2010cycle}
H.~Zhou and N.~Goertz.
\newblock Cycle analysis of time-invariant {LDPC} convolutional codes.
\newblock In {\em 2010 17th International Conference on Telecommunications},
  pages 23--28. IEEE, 2010.

\end{thebibliography}

\end{document}